\documentclass[11pt,envcountsect,oribibl]{llncs}

\usepackage{amsmath}
\usepackage{graphicx}
\usepackage{amssymb}
\usepackage{epstopdf}
\usepackage{url}
\usepackage{cite}
\usepackage{color}
\usepackage{algorithm, algorithmic}
\usepackage{subfigure}
\usepackage{fullpage}

\newcommand{\R}{\mathbb R}

\DeclareMathOperator{\pr}{pr}
\DeclareMathOperator{\ch}{conv}

\newcommand{\etal}{et al.}

\newcommand{\eps}{\varepsilon}
\newcommand{\eqdef}{:=}

\newtheorem{thm}{Theorem}[section]
\newtheorem{lem}[thm]{Lemma}
\newtheorem{observation}[thm]{Observation}

\newtheorem{cor}[thm]{Corollary}

\title{Approximating Tverberg Points in Linear Time for Any Fixed 
Dimension\thanks{A preliminary version appeared as
W. Mulzer and D. Werner, \emph{Approximating Tverberg Points in 
Linear Time for Any Fixed Dimension} in Proc.~28th SoCG, pp.~303--310, 2012.}}
\author{
Wolfgang Mulzer\inst{1} 
\and 
Daniel Werner\inst{2}\thanks{Funded by Deutsche Forschungsgemeinschaft 
within the Research Training
Group (Graduiertenkolleg) ``Methods for Discrete Structures''.} 
} 
\institute{
Institut f{\"ur} Informatik, Freie Universit{\"a}t Berlin, Germany,
\email{mulzer@inf.fu-berlin.de}, \url{http://page.mi.fu-berlin.de/mulzer} 
\and
Institut f{\"ur} Informatik, Freie Universit{\"a}t Berlin, Germany,
\email{dwerner@mi.fu-berlin.de}, \url{http://page.mi.fu-berlin.de/dawerner} 
}

\begin{document}

\maketitle

\begin{abstract}
Let $P \subseteq \R^d$ be a $d$-dimensional $n$-point set.
A \emph{Tverberg partition} is a partition of
$P$ into $r$ sets $P_1, \dots, P_r$ such
that the convex hulls $\ch(P_1), \dots, \ch(P_r)$ have non-empty intersection.
A point in $\bigcap_{i=1}^{r} \ch(P_i)$ is
called a \emph{Tverberg point} of depth $r$
for $P$. A classic result by Tverberg shows
that there always exists a Tverberg partition
of size $\lceil n/(d+1) \rceil$, but it is not
known how to find such a partition in polynomial time.
Therefore, approximate solutions are of interest.

We describe a deterministic algorithm that finds
a Tverberg partition of size $\lceil n/4(d+1)^3 \rceil$ in time
$d^{O(\log d)} n$. This means that for every fixed dimension
we can compute an approximate Tverberg point 
(and hence also an approximate \emph{centerpoint})
in \emph{linear} time. Our algorithm is obtained by
combining a novel lifting approach with a recent 
result by Miller and Sheehy~\cite{MillerSh10}.


\end{abstract}

\section{Introduction}\label{sec:introduction}

In many applications (such as statistical analysis
or finding sparse geometric separators in meshes) we would like to have a way 
to generalize the one-dimensional notion of a median to higher dimensions.
A natural way to do this uses the notion of \emph{halfspace depth}
(or \emph{Tukey depth}). 
\begin{definition} 
Let $P$ be a finite set of points in $\R^d$, and let $c \in \R^d$ be a point
(not necessarily in $P$). The \emph{halfspace depth of 
$c$ with respect to $P$} is 
\[ 
  \min_{\textnormal{halfspace } h,\, c \in h} |h \cap P|.
\]
The \emph{halfspace depth of $P$} is the maximum halfspace depth 
that any point $c \in \R^d$ can achieve.
\end{definition}

A classic result in discrete
geometry, the Centerpoint Theorem, claims that for every $d$-dimensional point set 
$P$ with $n$ points there exists
a \emph{centerpoint}, i.e., a point $c \in \R^d$ with halfspace depth
at least $n/(d+1)$~\cite{DanzerGrKl63,Rado46}. There are point sets where
this bound cannot be improved.

However, if we actually want to compute a centerpoint for
a given point set efficiently, the situation becomes more involved.
For $d=2$, a centerpoint
can be found deterministically in linear time~\cite{JadhavMu94}.
For general $d$, we can compute a centerpoint in 
$O(n^d)$ time using linear programming, since
Helly's theorem implies that the set of all centerpoints can be described as
the intersection
of $O(n^d)$ halfspaces~\cite{Edelsbrunner87}. 
Chan~\cite{Chan04} shows how to improve this running time to
$O(n^{d-1})$ with the help of randomization. He actually solves the 
apparently harder problem
of finding a point with maximum halfspace depth.
If the dimension is not fixed, a result by Teng 
shows that it is coNP-hard to check whether a given point is a 
centerpoint~\cite{Teng92}. 

However, as $d$ grows, a running time of $n^{\Omega(d)}$ is not feasible.
Hence, it makes sense to look for faster approximate solutions.
A classic approach uses $\eps$-approximations~\cite{Chazelle00}:
in order to obtain a point of halfspace depth $n(1/(d+1) - \eps)$, take a random 
sample $A \subseteq P$ of size $O((d/\eps^2) \log (d/\eps))$ and
compute a centerpoint for $A$ via linear-programming. This gives
the desired approximation with constant probability, and the 
running time after the sampling step is constant for fixed $d$. What more could we
possibly wish for? For one, the algorithm is Monte-Carlo:
with a certain probability, the reported point fails to be a centerpoint,
and we know of no fast algorithm to check its validity. This problem can
be solved by constructing the $\eps$-approximation 
deterministically~\cite{Chazelle00}, at
the expense of a more complicated algorithm. Nonetheless, in either case
the resulting running time grows exponentially 
with $d$, an undesirable feature for large dimensions.

This situation motivated Clarkson \etal~\cite{ClarksonEpMiStTe96} to
look for more efficient randomized algorithms for approximate centerpoints.
They give a simple probabilistic algorithm that computes a point of halfspace depth
$\Omega(n/(d+1)^2)$ in time $O(d^2(d \log n + \log(1/\delta))^{\log(d+2)})$, 
where $\delta$ is the error probability.
They also describe a more sophisticated algorithm
that finds such a point in time polynomial in $n$, $d$, and $\log(1/\delta)$.
Both algorithms are based on a repeated algorithmic application of Radon's 
theorem (see below).
Unfortunately, there remains a probability of $\delta$ that the result
is not correct, and we do not know how to detect failure efficiently.

\begin{figure}[t]
   \begin{center} 
     \includegraphics{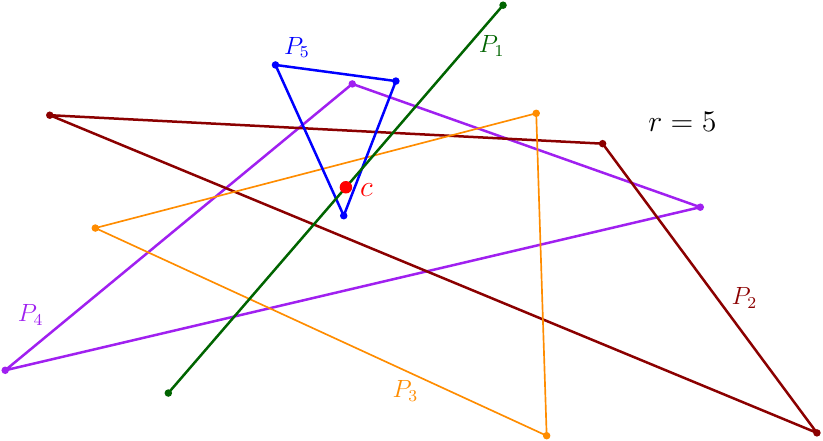}  
   \end{center}
     \caption{The point $c$ has Tverberg depth $r = 5$.}        
     \label{fig:TverbergPoint} 
\end{figure}

Thus, more than ten years later, Miller and Sheehy~\cite{MillerSh10} launched
a new attack on the problem. Their goal was to develop a deterministic
algorithm for approximating centerpoints whose running time is subexponential
in the dimension. For this, they use a different proof of the centerpoint
theorem that is based on a result by Tverberg: any $d$-dimensional 
$n$-point set can be partitioned into $r = \lceil n/(d+1) \rceil$ sets
$P_1, \dots, P_r$ such that the convex hulls $\ch(P_1), \dots, \ch(P_r)$ 
have nonempty intersection. Such a partition is called a \emph{Tverberg partition}
of $P$. By convexity, any point in $\bigcap_{i=1}^{r} \ch(P_i)$ must
be a centerpoint. 

More generally, we say that a point $c \in \R^d$ has \emph{Tverberg depth} 
$r'$ with respect to $P$ if there is a  partition of $P$ 
into $r'$ sets  such that  $c$ lies in the convex 
hull of each set. We also call $c$ an
\emph{approximate Tverberg point} (of depth $r'$); 
see Figure~\ref{fig:TverbergPoint}.

Miller and Sheehy describe how to find $\lceil n/2(d+1)^2\rceil$ disjoint subsets 
of $P$ and a point 
$c \in \R^d$ such that each subset contains $d+1$ points and has
$c$ in its convex hull. Hence, $c$ constitutes an approximate Tverberg point 
for $P$ (and thus also an approximate centerpoint), 
and the subsets provide a certificate for this fact. The algorithm is 
deterministic and runs in time $n^{O(\log d)}$. At the same time, it is 
the first algorithm that also finds an approximate Tverberg 
partition of $P$.  The running time is subexponential in $d$, but it 
is still the case that $n$ is raised to a power that depends on $d$, so
the parameters $n$ and $d$ are not separated in the running time. 

In this paper, we show that the running time for finding approximate 
Tverberg partitions (and hence approximate centerpoints) can be improved.
In particular, we show how to find a Tverberg partition with
$\lceil n/4(d+1)^3 \rceil$ sets in deterministic 
time $d^{O(\log d)} n$. This is linear in $n$ for any fixed dimension, 
and the dependence on $d$ is only quasipolynomial. 

\subsection{Some Discrete Geometry} 
We begin by recalling some basic facts and definitions from discrete
geometry~\cite{Mat02}. A classic fact about convexity is Radon's 
theorem. 
\begin{thm}[Radon's theorem] 
For any $P \subseteq \R^d$ with  $d+2$ points there exists a 
partition $(P_1, P_2)$ of $P$ such that 
$\ch(P_1) \cap \ch(P_2) \neq \emptyset$.
\end{thm} 
As mentioned above, Tverberg~\cite{Tverberg66} generalized this theorem
for larger point sets. 
\begin{thm}[Tverberg's theorem] 
Any set $P \subseteq \R^d$ with $n = (r-1)(d+1) + 1$ points 
can be partitioned into $r$ sets $P_1, \dots, P_r$  such that 
$\bigcap_{i=1}^r \ch(P_i) \neq \emptyset$.  
\end{thm}

Let $P$ be a set of $n$ points in $\R^d$. We say that $x \in \R^d$ has
\emph{Tverberg depth} $r$ (with respect to $P$) if there is a partition of 
$P$ into sets $P_1, \dots, P_r$ such that $x \in \bigcap_{i=1}^r \ch(P_i)$.  
Tverberg's theorem thus
states that, for any set $P$ in $\R^d$, there is a point of Tverberg depth 
at least
$\lfloor (n-1)/(d+1) + 1 \rfloor = \lceil n/(d+1)\rceil$. 
Note that every point with Tverberg depth $r$ also has halfspace depth
$r$.
Thus, from now on we will use the term \emph{depth} as a shorthand for 
Tverberg depth.
As remarked above, Tverberg's theorem immediately implies the famous
Centerpoint Theorem (see \cite{Mat02}):
\begin{thm}[Centerpoint Theorem]
For any set $P$ of $n$ points in $\R^d$ there is a point $c$
such that all halfspaces containing $c$ contain at least 
$\lceil n/(d+1) \rceil$ points from $P$.  
\end{thm} 

\noindent
Finally, another classic theorem will be useful for us.

\begin{thm}[Carath\'eodory's theorem]\label{thm:carat} 
Suppose that $P$ is a set of $n$ points 
in $\R^d$ and $x \in \ch(P)$. Then there is
a set of $d+1$ points $P' \subseteq P$ such that $x \in \ch(P')$.  
\end{thm} 

This means
that, in order to describe a Tverberg partition of depth $r$, we 
need only $r (d+1)$ points from $P$.  This observation is also 
used by Miller and Sheehy~\cite{MillerSh10}. They further note that 
it takes $O(d^3)$ time to
replace $d+2$ points by $d+1$ points by using 
Gaussian elimination.  We denote the process of replacing
larger sets by sets of size $d+1$ as \emph{pruning}, see Lemma~\ref{lem:pruning}.

\subsection{Our Contribution} 
We now describe our results in more detail. 
In Section~\ref{sec:simpleAlgorithm}, we present a
simple lifting argument which leads to an easy Tverberg approximation
algorithm. 

\begin{thm}\label{thm:simpleAlgorithm} 
Let $P$ be a set of $n$ points in
$\R^d$ in general position. One can compute a Tverberg point of depth 
$\lceil n/2^d \rceil$ for $P$ and the corresponding partition in 
time $d^{O(1)} n$.  
\end{thm} 

While this does not yet give a good approximation ratio (though constant 
for any fixed $d$), it is a natural approach to the problem: it 
computes a higher dimensional Tverberg point via successive median 
partitions---just as a Tverberg point is a higher dimensional generalization 
of the $1$-dimensional median.

By collecting several low-depth points and afterwards applying the brute-force 
algorithm on small point sets, we get an even higher depth in linear time for 
any \emph{fixed} dimension:

\begin{thm}\label{thm:betterTverberg} 
Let $P$ be a set of $n$ points in $\R^d$. Then one can find a 
Tverberg point of depth $\lceil n/2(d+1)^2 \rceil$ and
a corresponding partition in time $f(2^{d+1}) + d^{O(1)} n$, where $f(m)$ is 
the time to compute a Tverberg point of depth $\lceil m/(d+1) \rceil$ for 
$m$ points by brute force.
\end{thm} 

Finally, by combining our approach with that of Miller and Sheehy, we can improve 
the running time to be quasipolynomial in $d$:
\begin{thm}\label{thm:bootstrap} 
Let $P$ be a set of $n$ points in
$\R^d$. Then one can compute a Tverberg point of depth 
$\lceil n/4(d+1)^3 \rceil$ and a corresponding pruned partition in 
time $d^{O(\log d)} n$.  
\end{thm}

\noindent
In Section \ref{sec:comparison}, we compare these results to the Miller-Sheehy 
algorithm and its extensions.

\section{A Simple Fixed-Parameter Algorithm}\label{sec:simpleAlgorithm} 

We now present a simple algorithm that runs in linear time for any fixed 
dimension and computes a point of depth $\lceil n/2^d \rceil$. For this, we 
show how to compute a Tverberg point by recursion on the dimension. As a 
byproduct, we obtain a quick proof of a weaker version of Tverberg's theorem.
First, however, we give a few more details about the basic operations
performed by our algorithm.

\subsection{Basic Operations}
Our algorithm builds a Tverberg partition for a $d$-dimensional point set
$P$ by recursion on the dimension. In each step, we store a 
Tverberg partition for some point set, together with an approximate Tverberg point $c$.
We have for each set $P_i$ in the partition a convex combination that witnesses
$c \in \ch(P_i)$. 
All the points that arise during our algorithms are obtained 
by repeatedly taking convex combinations of the input points, so 
the following simple observation lets us maintain this invariant.
\begin{observation} 
If $x_i = \sum_{p \in P_i} \alpha_p p$ and $y = \sum_i \beta_i x_i$ are 
convex combinations, then \[ y = \sum_i \sum_{p \in P_i} \beta_i \alpha_p p \] 
is a convex combination of the set $\bigcup P_i$ for $y$. \qed
\end{observation}

By Carath\'eodory's theorem (Theorem~\ref{thm:carat}), a Tverberg 
partition of depth $r$ can be described by $r (d+1)$ points from $P$. In order 
to achieve running time $O(n)$, we need the following 
observation, also used by Miller and Sheehy~\cite{MillerSh10}.

\begin{lem}\label{lem:pruning} 
Let $Q \subseteq \R^d$ be a set of $m \geq d + 2$ points with $c \in \ch(Q)$, 
and suppose we have a convex combination of $Q$ for $c$. 
Then we can find a subset $Q' \subset Q$ with $d+1$ points such
that $c \in \ch(Q')$, together with a corresponding 
convex combination, in 
time $O(d^3 m)$.
\end{lem}

\begin{proof} 
Miller and Sheehy observe that replacing $d+2$ points by $d+1$ points takes 
$O(d^3)$ time by finding an affine dependency through Gaussian elimination, 
see Gr\"otschel, Lov{\'a}sz, and Schrijver~\cite[Chapter 1]{GroetschelLoSc93}. 
The choice of affine dependencies does not matter. Thus, in order to eliminate
a point from $Q$, we can take any 
subset of size $d+2$, resolve one of the affine dependencies,
and update the convex combination accordingly. Repeating this process, we 
can replace $m$ points by $d+1$ points in time $(m - (d+1))O(d^3) = O(d^3m)$.
\qed\end{proof}

The process in Lemma~\ref{lem:pruning} is called
\emph{pruning}, and we call a partition of a $d$-dimensional point set in 
which all sets have size at most $d + 1$ a \emph{pruned partition}. 
This will enable us to bound the cost of many operations in terms of the 
dimension $d$, instead of the number of points $n$.

\subsection{The Lifting Argument and a Simple Algorithm} 
Let $P$ be a $d$-dimensional point set.
As a Tverberg point is a higher dimensional version of the median, a 
natural way to compute a Tverberg point for $P$
is to first project $P$ to some lower-dimensional space,
then to recursively compute a good Tverberg point for this projection, 
and to use this point to find a solution in the higher-dimensional 
space. Surprisingly, we are not aware of any such argument 
having appeared in the literature so far.

In what follows, we will describe how to \emph{lift} a lower-dimensional 
Tverberg point into some higher dimension. Unfortunately, this process
will come at the cost of a decreased depth for the lifted Tverberg point.
For clarity of presentation, we first explain the lifting lemma in its 
simplest form. 
In Section~\ref{sec:generalLiftingLemma}, we then state the lemma 
in its full generality.

\begin{figure}[ht] 
  \begin{center}
  \subfigure[project]{ 
    \includegraphics[scale=.4]{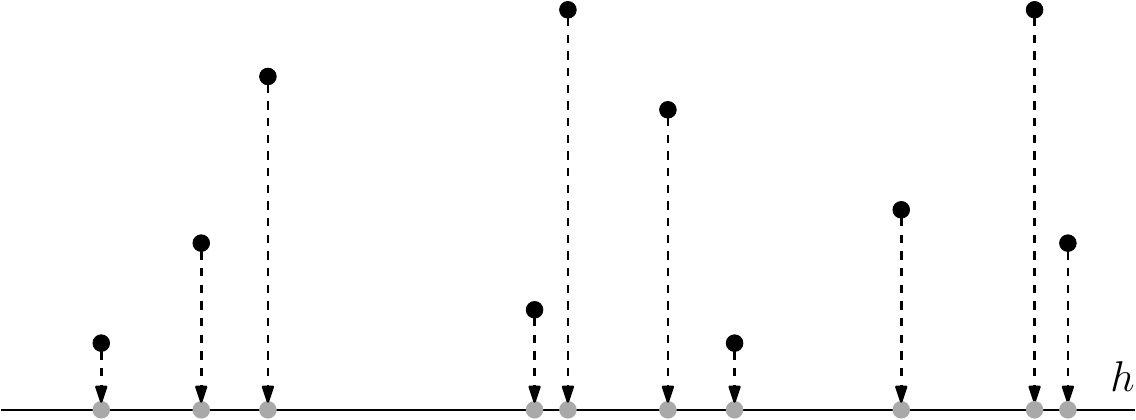}
   } 
  \subfigure[find partition]{ 
    \includegraphics[scale=.4]{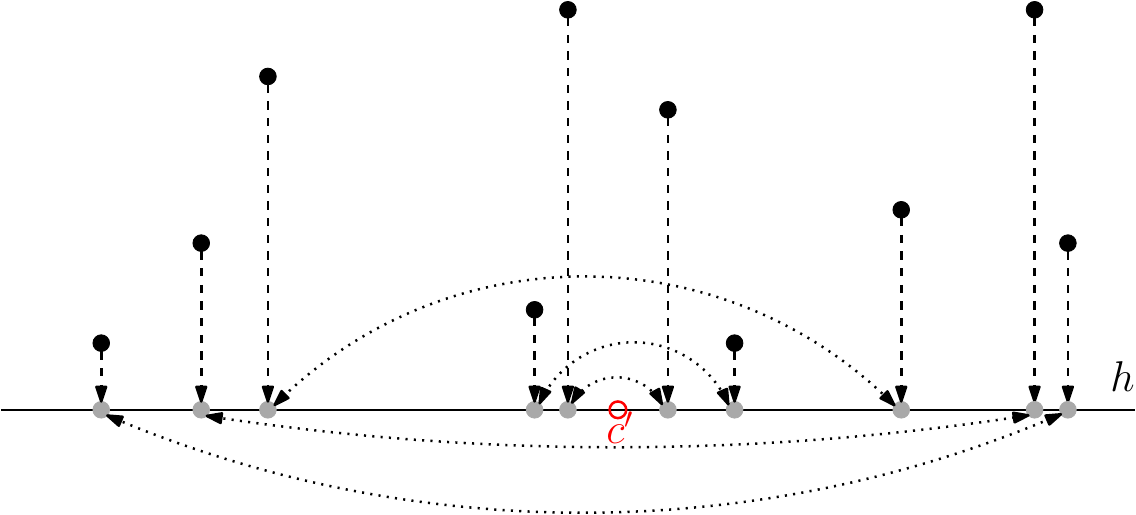}

   } 
  \subfigure[intersect hulls of the sets with $h^{\perp}$]{ 
    \includegraphics[scale=.4]{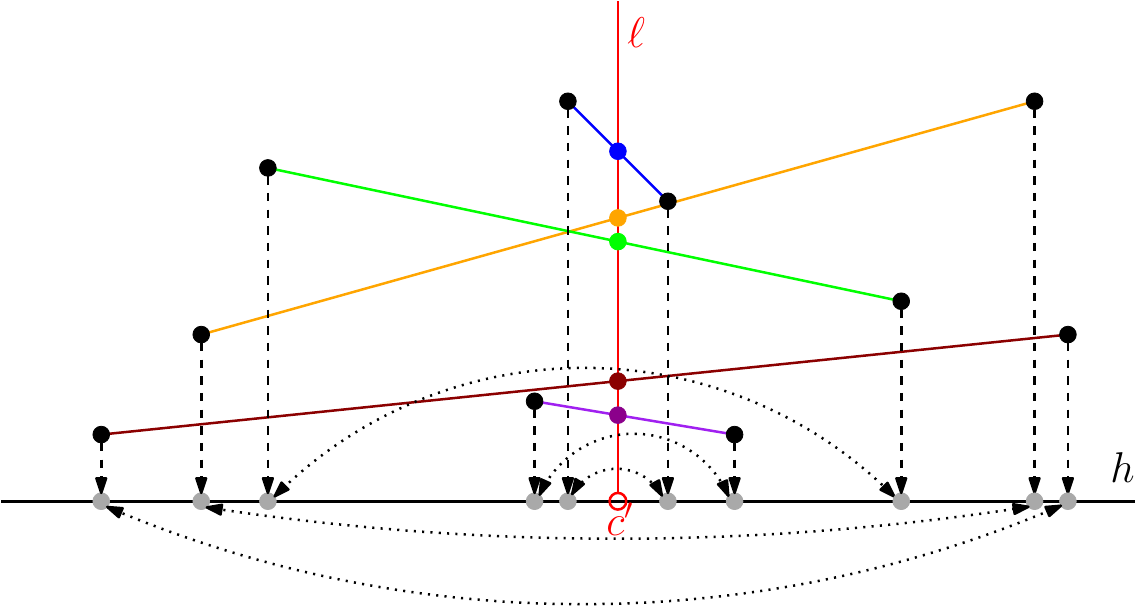}
   } 
  \subfigure[find median of intersections and combine]{ 
    \includegraphics[scale=.4]{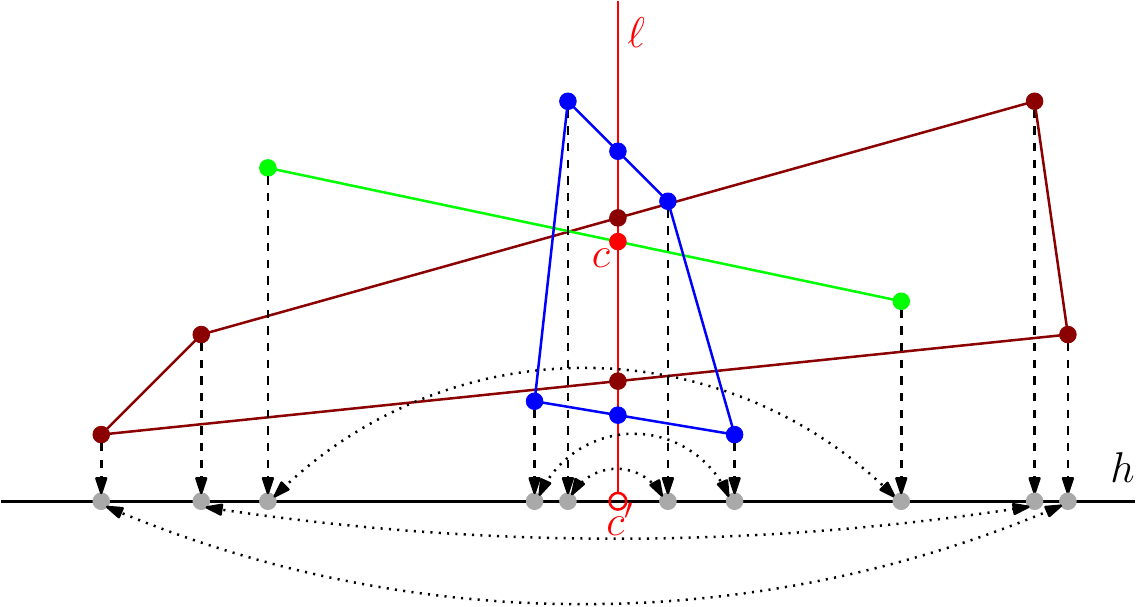}
   } 
 \end{center}
     \caption{Illustrating the lifting lemma in the plane: we project the point
       set $P$ to the line $h$ and find a Tverberg partition and a Tverberg
       point $c'$ for the projection. Then, we construct the line $\ell$ through 
       $c'$ that is perpendicular to $h$, and we take the intersection 
       with the lifted convex hulls of the Tverberg partition.
       We then find the median $c$ and the corresponding partition for the 
       intersections along $\ell$. Finally, we group the points according to this
       partition.}        
     \label{fig:Example} 
\end{figure}

\begin{lem}\label{lem:lifting}  
Let $P$ be a set of $n$ points in $\R^d$, and let
$h$ be a hyperplane in $\R^d$. Let $c' \in h$ be a Tverberg point of 
depth $r$ for the projection of $P$ onto $h$, with pruned partition 
$P_1, \dots, P_r$. Then we can find a Tverberg 
point $c \in \R^d$ of depth $\lceil r/2 \rceil$ for $P$ and a corresponding
Tverberg partition in time $O(dn)$. 
\end{lem} 

\begin{proof} 
For every point $p \in P$, let $\pr(p)$ denote the projection of
$p$ onto $h$, and for every $Q \subseteq P$, let $\pr(Q)$ be the projections
of all the points in $Q$. Let $P_1, \dots, P_r \subseteq P$ 
such that $\pr(P_1), \dots, \pr(P_r)$ is a pruned partition
for $\pr(P)$ with Tverberg point $c'$.
Let $\ell$ be the line through $c'$ orthogonal to
$h$. 

Since our assumption implies $c' \in \ch(\pr(P_i))$ for $i=1, \dots, r$, it 
follows that $\ell$ intersects each $\ch(P_i)$ at some point 
$x_i \in \R^d$. More precisely, as we have a convex combination 
$c' = \sum_{p \in P_i} \alpha_p \pr(p)$ for each $P_i$, we 
simply get $x_i = \sum_{p \in P_i} \alpha_p p$.

Assuming an appropriate numbering, let 
$\widehat Q_i = \left\{ x_{2i - 1}, x_{2i} \right\}$, 
$i = 1, \dots, \lceil r/2 \rceil$, be a Tverberg partition 
of $x_1, \dots, x_r$. (If $r$ is odd, the set $\widehat Q_{\lceil r/2 \rceil}$ contains 
only one point, the median.) Since the points $x_i$ lie on the 
line $\ell$, such a Tverberg partition exists and can be computed 
in time $O(r)$ by finding the median $c$, i.e., the element of rank 
$\lceil r/2 \rceil$, according to the order along $\ell$ (see 
Cormen~\etal~\cite[Chapter~9]{CormenLeRiSt09}).
We claim that $c$ is a Tverberg point for $P$ of depth 
$\lceil r/2 \rceil$.  Indeed, we have 
\[ 
c \in \ch(\widehat Q_i) = 
\ch(\left\{ x_{2i - 1}, x_{2i} \right\}) \subseteq \ch(P_{2i - 1} \cup P_{2i}),
\] 
for $1 \leq i \leq \lceil r/2 \rceil$. Thus, if we set 
$Q_i := P_{2i - 1} \cup P_{2i}$, then
$Q_1, \dots, Q_{\lceil r/2 \rceil}$ is a Tverberg partition
for the point $c$.
The total time to compute $c$ and the $Q_i$ is $O(n)$,
as claimed. See Figure~\ref{fig:Example} for a two-dimensional
illustration of the lifting argument.
\qed\end{proof}

\noindent
Theorem~\ref{thm:simpleAlgorithm} is now a direct consequence
of Lemma~\ref{lem:lifting}.

\begin{thm}[Thm~\ref{thm:simpleAlgorithm}, restated]
Let $P$ be a set of $n$ points in
$\R^d$ in general position. One can compute a Tverberg point of depth 
$\lceil n/2^d \rceil$ for $P$ and the corresponding partition in 
time $d^{O(1)} n$.  
\end{thm}

\begin{proof}
If $d = 1$, we obtain a Tverberg point and a corresponding
partition by finding the median $c$ of $P$~\cite{CormenLeRiSt09} and pairing 
each point to the left of $c$ with exactly one point to the right 
of $c$.

If $d > 1$, we project $P$ onto the hyperplane $x_d = 0$.
This gives an $n$-point set $P' \subseteq \R^{d-1}$. We 
recursively find a point of depth 
$\lceil n/{2^{d-1}} \rceil$ and a corresponding pruned partition
for $P'$. We then apply 
Lemma~\ref{lem:lifting} to get a point $c \in \R^d$ of depth
$r = \left\lceil \lceil n/{2^{d-1}} \rceil / 2 \right\rceil \geq \lceil n/2^d \rceil$ 
for $P$,
together with a partition. Each set has at
most $2d$ points, so by applying Lemma~\ref{lem:pruning}  to
each set, it takes 
$O(d^4r)$ time to prune all sets.

This yields a total running time of $T_d(n) \leq T_{d-1}(n) + d^{O(1)}n$, 
which implies the result.
\qed\end{proof}

\noindent
In particular, Theorem~\ref{thm:simpleAlgorithm} gives a weak version of 
Tverberg's theorem with a 
simple proof.

\begin{cor}[Weak Tverberg theorem] 
Let $P$ be a set of $n$ points in $\R^d$.
Then $P$ can be partitioned into $\lceil n/2^d \rceil$ sets 
$P_1, \dots, P_{\lceil n/2^d \rceil}$ such that
\[ 
\bigcap_{i=1}^{\lceil n/2^d \rceil} \ch(P_i) \neq \emptyset. 
\] \qed
\end{cor}

\subsection{An Improved Approximation Factor}

In order to improve the approximation factor, we will now use an easy 
bootstrapping approach.
A Tverberg partition of depth $r$ in $\R^d$ needs only $(d+1)r$ points. 
This means that after finding
a point of depth $n/2^d$, we still have
$n\left(1 - (d+1)/2^d\right)$ unused points at our disposal. 
The next lemma shows how to leverage these points to achieve an even
higher Tverberg depth.

\begin{lem}\label{lem:collection} 
Let $\rho \geq 2$ and $q(m,d)$ be a function such that
for any $m$-point set $Q \subseteq \R^d$ we can compute a 
point of depth $\lceil m/\rho \rceil$ and a corresponding 
pruned partition in time $q(m, d)$. 

Let $P \subseteq \R^d$ with $|P|=n$, and let $\beta \in [2, n/\rho]$ be 
a constant. Define the target depth $\delta$ as $\delta \eqdef \lceil n/\beta\rho \rceil$.
Then we can find $\alpha := \lceil \frac{n(1 - 1/\beta)}{\delta (d+1)} \rceil$ 
disjoint subsets $Q_1, \dots, Q_\alpha$ of $P$ such that for each $Q_i$ we have 
a Tverberg point $c_i$ of 
depth $\delta$ and a pruned partition $\mathcal Q_i$. 
This takes total time 
\[
O\Biggl(\frac{(\beta-1)\rho \,q(n,d)}{d+1}
\Biggr).
\] 
\end{lem} 

\begin{proof}
Let $P_1 \eqdef P$.
We take an arbitrary subset $P_1' \subseteq P_1$ with $\lceil n/\beta \rceil$ 
points and find a Tverberg point $c_1$ of depth $\delta$ and a corresponding
pruned partition $\mathcal{Q}_1$ for $P'_1$. 
This takes time $q(n,d)$, and the set 
$Q_1 \eqdef \bigcup_{Z \in \mathcal{P}_1} Z$ contains at most $\delta(d+1)$
points. Set $P_2 \eqdef P_1 \setminus Q_1$ and repeat. 
The resulting sets $Q_i$ are
pairwise disjoint, and we can repeat this process until
\[
	n - i \delta (d+1) <  \frac{n}{\beta}.
\]
This gives
\[
	\alpha \geq i > \left\lceil \frac{n(1 - 1/\beta)}{\delta (d+1)} \right\rceil.
\]
Thus, we obtain $\alpha$ points $c_1, \dots, c_\alpha$ with corresponding 
Tverberg partitions $\mathcal{Q}_1, \dots, \mathcal{Q}_\alpha$, each of 
depth at least $\lceil n/\beta\rho \rceil$, as desired.
\qed\end{proof}

For example, by Theorem~\ref{thm:simpleAlgorithm} we can find a point of depth 
$\lceil n/2^d \rceil$ and a corresponding pruned partition in time $d^{O(1)}n$. Thus,
by applying Lemma~\ref{lem:collection} with $c = 2$, $\rho = 2^d$, we can also 
find $\lceil n/(2\lceil n/2^{d+1} \rceil(d+1)) \rceil \approx  2^d/(d+1) $ 
points of depth $\lceil n/2^{d+1} \rceil$ in linear time.

In order to make use of Lemma~\ref{lem:collection}, we will also need a 
lemma that describes how we can combine these points in order to increase 
the total depth. This generalizes a similar lemma by Miller and 
Sheehy~\cite[Lemma~4.1]{MillerSh10}.  

\begin{lem}\label{lem:combine}
Let $P$ be a set of $n$ points in
$\R^d$, and let $P = \biguplus_{i=1}^\alpha P_i$ be a partition of $P$. 
Furthermore, suppose that for each $P_i$ we have a Tverberg point 
$c_i \in \R^d$ of depth $r$, together with a corresponding pruned
partition $\mathcal{P}_i$. Let 
$C \eqdef \{ c_i \mid 1 \leq i \leq \alpha\}$ and $c$ be a 
point of depth $r'$ for $C$, with corresponding pruned 
partition $\mathcal{C}$. Then $c$ is a point of depth $r
r'$ for $P$.  Furthermore, we can find a corresponding pruned
partition in time $d^{O(1)}n$.
\end{lem} 

\begin{proof}
For $i = 1, \dots, \alpha$, write $\mathcal{P}_i = \{Q_{i1}, \dots, Q_{ir}\}$,
and write $\mathcal{C} = \{D_1, \dots, D_{r'}\}$. For $a = 1, \dots, r'$, 
$b = 1, \dots, r$, we define sets $Z_{ab}$ as 
\[
Z_{ab} \eqdef \bigcup_{c_i \in D_a} Q_{ib}.
\]
 We claim that the set $\mathcal{Z} \eqdef
\{Z_{ab} \mid a = 1, \dots, r'; b = 1, \dots, r\}$ is a Tverberg
partition of depth $r r'$ for $P$ with Tverberg point $c$.
By definition, $\mathcal{Z}$ is a partition with the appropriate
number of elements. It only remains to check that $c \in \ch(Z_{ab})$ for
each $Z_{ab}$. Indeed, we have
\[
c \in \ch(D_a) = \ch\left(\bigcup_{c_i \in D_a} \left\{ c_i \right\} \right) 
\subseteq \ch\Bigl(\bigcup_{c_i \in D_a} \ch\left(Q_{ib}\right)\Bigr)  
= \ch\Bigl(\bigcup_{c_i \in D_a} Q_{ib}\Bigr) = \ch(Z_{ab}), 
\]
for $a = 1 \dots r', b = 1 \dots r$.

As the partitions $\mathcal P_i$ and $\mathcal C$ were pruned, 
each $Z_{ab}$ consists of at most $(d+1)^2$ points. Thus, by 
Lemma~\ref{lem:pruning}, each $Z_{ab}$ can be pruned in time 
$O(d^5)$. Since $|\mathcal{Z}| \leq  n$, the lemma 
follows.
\qed\end{proof}

\noindent
Combining Lemmas~\ref{lem:collection} and~\ref{lem:combine}, we can now prove
Theorem~\ref{thm:betterTverberg}.

\begin{theorem}[Thm~\ref{thm:betterTverberg}, restated]
Let $P$ be a set of $n$ points in $\R^d$. Then one can find a 
Tverberg point of depth $\lceil n/2(d+1)^2 \rceil$ and
a corresponding partition in time $f(2^{d+1}) + d^{O(1)} n$, where $f(m)$ is 
the time to compute a Tverberg point of depth $\lceil m/(d+1) \rceil$ for 
$m$ points by brute force, together with an associated Tverberg partition.
\end{theorem}
\begin{proof}
If $n \leq 2^{d+1}$, we solve the problem by brute-force in 
$f(2^{d+1})$ time.
Otherwise, we apply Lemma~\ref{lem:collection} with $c = 2$ and $\rho = 2^d$ 
to obtain a set $C$ of 
\[ 
|C|  = \left\lceil \frac{n}{2\lceil n/2^{d+1} \rceil(d+1)} \right\rceil 
\] 
points of depth $\lceil n/2^{d+1} \rceil$ for $P$  with corresponding pruned 
partitions in time $d^{O(1)}n$.
We then use the brute-force algorithm to get a Tverberg point
for $C$ with depth $\lceil |C|/(d+1) \rceil$ and a corresponding partition, in time
$f(|C|)$. Finally, we apply Lemma~\ref{lem:combine}
to obtain a Tverberg point and corresponding partition in  time $d^{O(1)} n$.
Using that $\lceil a \lceil b \rceil\rceil \geq \lceil a b \rceil$ and
$\lceil a \rceil \lceil b \rceil \geq \lceil a \lceil b \rceil \rceil$ for $a,b \geq 0$, 
we get that the resulting depth is
\[
	\left\lceil \frac{n}{2^{d+1}} \right\rceil \cdot 
	\left\lceil \frac{|C|}{d+1} \right\rceil 
	\geq  \left\lceil \left\lceil \frac{n}{2^{d+1}} \right\rceil 
	\frac{n}{2\lceil n/2^{d+1} \rceil(d+1)^2} \right\rceil
  = \left\lceil \frac{n}{2(d+1)^2} \right\rceil,
\]
and the total running time is $f(2^{d+1}) + d^{O(1)}n$, as desired.
\qed\end{proof}

Instead of brute force, we can also use
the algorithm by Miller and Sheehy to find a point among the deep 
points. This gives a worse depth, but it is slightly faster.

\begin{theorem}\label{thm:betterTverberg-withMS} 
Let $P$ be a set of $n$ points in $\R^d$. Then one can compute a 
Tverberg point of depth $\lceil n/4(d+1)^3 \rceil$ and
a corresponding partition in time $2^{O(d \log d)} + d^{O(1)}n$.
\end{theorem} 

\begin{proof}
For $n \leq 2^{d+1}$ we use the Miller-Sheehy algorithm to get a 
point of depth $\lceil n/2(d+1)^2\rceil$ in time $2^{O(d \log d)}$.
Otherwise, we proceed as in the proof of Theorem~\ref{thm:betterTverberg}
to obtain a set $C$ of 
\[ 
|C|  = \left\lceil \frac{n}{2\lceil n/2^{d+1} \rceil(d+1)} \right\rceil 
\] 
Tverberg points of depth
$\lceil n/2^{d+1}\rceil$ and corresponding pruned partitions in time $d^{O(1)}n$.
The Miller-Sheehy algorithm then gives a Tverberg point for $C$ of depth
$\lceil |C|/2(d+1)^2\rceil$ in time $|C|^{O(\log d)} = 2^{O(d \log d)}$.
Finally, we apply Lemma~\ref{lem:combine}. This takes time $d^{O(1)}n$ and yields
a Tverberg point and pruned partition of depth
\[
	\left\lceil \frac{n}{2^{d+1}} \right\rceil 
	\cdot \left\lceil \frac{|C|}{2(d+1)^2} \right\rceil 
	\geq  \left\lceil \left\lceil \frac{n}{2^{d+1}} \right\rceil 
	\frac{n}{4\lceil n/2^{d+1} \rceil(d+1)^3} \right\rceil
  = \left\lceil \frac{n}{4(d+1)^3} \right\rceil,
\]
as claimed.
\qed\end{proof}

\section{An Improved Running Time}\label{sec:improvedAlgorithm} 
The algorithm from the previous section runs in linear time for any 
fixed dimension, but the constants are huge. In this section, we show 
how to speed up our approach through an improved recursion,
and we obtain an algorithm with running time $d^{O(\log d)} n$ while 
losing a depth factor of $1/2(d+1)$.

\subsection{A More General Version of the Lifting Argument}
\label{sec:generalLiftingLemma} 

We first present a more general
version of the lifting argument in
Lemma~\ref{lem:lifting}. For this, we need some more notation. 
Let $P \subseteq \R^d$ be finite. A $k$-dimensional \emph{flat} 
$F \subseteq \R^d$ (often abbreviated as \emph{$k$-flat}) is defined as 
a $k$-dimensional affine subspace of $\R^d$ (or, equivalently,
as the affine hull of $k+1$ affinely independent points in $\R^d$). 
We call a $k$-dimensional flat $F \subseteq \R^d$ a \emph{Tverberg $k$-flat 
of depth $r$ for $P$} if there is a 
partition of $P$ into sets $P_1, \dots, P_r$ such that $\ch(P_i)
\cap F \neq \emptyset$ for all $i = 1, \dots, r$. This generalizes the 
notion of a Tverberg point.

\begin{lem}\label{lem:liftingExtended} 
Let $P$ be a set of $n$ points in $\R^d$, and let $h \subseteq \R^d$ be a
$k$-flat. 
Suppose we have a Tverberg point $c \in h$ of depth $r$ for 
$\pr(P) := \pr_h(P)$, as well as a corresponding Tverberg partition. 
Let $h^{\perp}_c$ be the $(d-k)$-flat orthogonal to $h$ that passes 
through $c$. Then $h^{\perp}_c$ is a Tverberg $(d-k)$-flat for
$P$ of depth $r$, with the same Tverberg partition.  
\end{lem} 

\begin{proof} 
Let $\pr(P_1), \dots, \pr(P_r)$ be the Tverberg partition 
for the projection $\pr(P)$.
It suffices  to show that $\ch(P_i)$ intersects $h^{\perp}_c$ for 
$i = 1, \dots, r$. Indeed, for $P_i = \{p_{i1}, \dots, p_{il_i}\}$
let $c = \sum_{j=1}^{l_i} \lambda_j \pr(p_{ij})$ be a convex combination 
that witnesses $c \in \ch(\pr(P_i))$. We now write each 
$p_{ij} = \pr(p_{ij}) + \pr^{\perp}(p_{ij})$, where $\pr^{\perp}(\cdot)$ 
denotes the projection onto the orthogonal complement $h^{\perp}$ of $h$.
Then,
\[
\sum_{j=1}^{l_i} \lambda_j  p_{ij} =
\sum_{j=1}^{l_i} \lambda_j \pr(p_{ij}) + \sum_{j=1}^{l_i} \lambda_j
\pr^{\perp}(p_{ij}) \in c + h^{\perp} = h^{\perp}_c,
\] 
as claimed.
\qed\end{proof}

\noindent
Lemma~\ref{lem:liftingExtended} lets us use a good algorithm for 
\emph{any fixed dimension} 
to improve the general case. 

\begin{lem}\label{lem:lowDtoHighD} 
Let $\delta \geq 1$ be a fixed integer.
Suppose we have an algorithm $\mathcal{A}$ with the following property:
for every point set $Q \subseteq \R^\delta$, the algorithm 
$\mathcal{A}$ constructs a Tverberg point of depth $\lceil |Q|/\rho \rceil$ 
for $Q$ as well as a corresponding pruned partition in time $f(|Q|)$. 

Then, for any $n$-point set $P \subseteq \R^d$ and for any $d \geq \delta$, 
we can find a Tverberg point of depth $\lceil n/\rho^{\lceil d/\delta \rceil} \rceil$ and 
a corresponding pruned partition in time 
$\lceil d/\delta \rceil f(n) + d^{O(1)}n$.
\end{lem}

\begin{proof}
Set $k \eqdef \lceil d/\delta \rceil$. 
We use induction on $k$ to show that such an algorithm exists
with running time $k(f(n) + d^{O(1)}n)$.  If $k = 1$, we can just 
use algorithm $\mathcal{A}$, and there is nothing to show.

Now suppose $k > 1$.  Let $h \subseteq \R^d$ be a $\delta$-flat in $\R^d$, 
and let $\pr(P)$ be the projection of $P$ onto $h$. We use algorithm
$\mathcal{A}$ to find a Tverberg point $c$ of depth $\lceil n/\rho \rceil$ for $\pr(P)$ 
as well as a corresponding pruned partition 
$\pr(P_1), \dots, \pr(P_{\lceil n/\rho \rceil})$. This takes time $f(n)$. By 
Lemma~\ref{lem:liftingExtended}, the $(d-\delta)$-flat $h^{\perp}_c$
is a Tverberg flat of depth $\lceil n/\rho \rceil$ for $P$, with corresponding pruned
partition $P_1, \dots, P_{\lceil n/\rho \rceil}$. For each $i$, we can thus
find a point $q_i$ in  $\ch(P_i) \cap h_c^{\perp}$ in
time $d^{O(1)}$.

Now consider the point set $Q = \{q_1, \dots, q_{\lceil n/\rho \rceil}\} 
\subseteq h_c^{\perp}$. The set $Q$ is $(d-\delta)$-dimensional. 
Since $\lceil (d-\delta)/\delta \rceil = k-1$,  we can inductively find
a Tverberg point $c'$  for $Q$ of depth $\lceil |Q|/\rho^{\lceil d/\delta \rceil - 1}\rceil 
\geq \lceil n/\rho^{\lceil d/\delta \rceil} \rceil$ and a
corresponding pruned partition 
$\mathcal{Q}$ in total time $(k-1)(f(n) + d^{O(1)} n)$.
Now, $c'$ is a Tverberg point of depth $n/\rho^{\lceil d/\delta \rceil}$ for
$P$: a corresponding Tverberg partition is obtained by replacing each
point $q_i$ in the partition $\mathcal{Q}$ by the corresponding subset
$P_i$. The resulting partition can be pruned in time $d^{O(1)}n$.
Thus, the total running time is 
\[
(k-1)(f(n) + d^{O(1)} n) + f(n) + d^{O(1)}n = k(f(n) + d^{O(1)}n),
\]
and since $k = O(d)$, the claim follows.
\qed\end{proof}

In the journal version of this paper, we claimed an example
application of 
Lemma~\ref{lem:lowDtoHighD}
to obtain a point of 
depth $n/4^{\lceil d/3 \rceil}$ 
in time $O(n \log n + d^{O(1)}n)$.
However, this example was based on an incorrect citation
and the conclusion does not follow as claimed. 

\subsection{An Improved Algorithm} 

Finally, we show how to combine the above techniques 
to obtain an algorithm with a better running time. 
The idea is as follows: using 
Lemma~\ref{lem:lowDtoHighD}, we can reduce one $d$-dimensional
instance to two instances of dimension $d/2$. We would like to proceed
recursively, but unfortunately, this
reduces the depth of the partition. To fix this,
we apply Lemmas~\ref{lem:collection}, \ref{lem:combine} and the 
Miller-Sheehy algorithm. 

\begin{thm}[Thm.~\ref{thm:bootstrap}, restated]
Let $P$ be a set of $n$ points in
$\R^d$. Then one can compute a Tverberg point of depth 
$\lceil n/4(d+1)^3 \rceil$ and a corresponding pruned partition in 
time $d^{O(\log d)} n$.  
\end{thm}

\begin{proof}
We prove the theorem by induction on $d$.
As usual, for $d=1$ the problem reduces to median computation, 
and the result is immediate. 

Now let $d \geq 2$. 
By induction, for any at most $\lceil d/2 \rceil$-dimensional point set 
$Q \subseteq \R^{\lceil d/2 \rceil}$
there exists an algorithm that returns a Tverberg point of depth
$\lceil |Q|/4(\lceil d/2 \rceil + 1)^3 \rceil$ and a corresponding pruned 
partition in time $d^{\alpha \log \lceil d/2 \rceil} n$, for some sufficiently 
large constant $\alpha > 0$. 

Thus, by Lemma~\ref{lem:lowDtoHighD} (with $\delta = \lceil d/2 \rceil$), 
there exists an algorithm that can compute a Tverberg point for $P$ of 
depth $\lceil n/16(\lceil d/2 \rceil + 1)^6 \rceil$ and a corresponding 
Tverberg partition in total time $2d^{\alpha \log \lceil d/2 \rceil } + d^{O(1)}n$.
Now we apply Lemma~\ref{lem:collection} with $c = 2$ and 
$\rho = 16(\lceil d/2 \rceil + 1)^6$. The lemma shows that
we can compute a set $C$ of  $\lceil 16 (\lceil d/2 \rceil + 1)^6 / (d+1) \rceil$ 
points of depth
$\delta = \lceil n/32(\lceil d/2 \rceil + 1)^6 \rceil$ and corresponding 
(disjoint) pruned partitions in time
$d^{\alpha \log \lceil d/2 \rceil + O(1)}n$.
Applying the Miller-Sheehy
algorithm, we can find a Tverberg point for $C$ of depth
$\lceil |C|/2(d+1)^2 \rceil$ and a corresponding pruned partition in time
$|C|^{O(\log d)}$. Now, Lemma~\ref{lem:combine} shows that in additional
$d^{O(1)}n$ time, we obtain a Tverberg point and a corresponding Tverberg
partition for $P$ of size
\[
  \left\lceil \frac{n}{2 \cdot 16 (\lceil d/2 \rceil + 1)^6} \right\rceil 
  \left\lceil \frac{ 16 (\lceil d/2 \rceil + 1)^6 }{ 2(d+1)^2 
  (d+1)} \right\rceil \geq \left\lceil \frac{n}{4(d+1)^3} \right\rceil, 
\] 
since $\lceil a \rceil \lceil b \rceil \geq \lceil ab\rceil$ for all
$a,b \geq 0$.

It remains to analyze the running time.
Adding the various terms, we obtain a time bound of
\[ 
T(n,d) = d^{\alpha \log \lceil d/2 \rceil +O(1)}n + |C|^{O(\log d)} + d^{O(1)}n.
\] 
Since $|C| = d^{O(1)}$, using 
$\log \lceil d/2 \rceil \leq \log (d/2) + \log (2\lceil d/2 \rceil / d) \leq
\log (d/2) +  \log (4/3)$, we get
\begin{align*}
T(n,d) & \leq  d^{\alpha \log \lceil d/2 \rceil +O(1)}n + d^{O(\log d)}n \\
	   & \leq  d^{\alpha \log d - \alpha/2}n + d^{\beta \log d}n,
\end{align*}
for $\alpha$ large enough and some $\beta > 0$, independent of $d$.
Hence,  for large enough $\alpha$ we have
\[ 
T(n,d) \leq d^{\alpha \log d}n = d^{O(\log d)}n,
\] 
as claimed.
This completes the proof.  
\qed\end{proof}

Thus, we can compute a polynomial approximation to a Tverberg point 
in time pseudopolynomial in $d$ and linear in $n$.

\section{Comparison to Miller-Sheehy}\label{sec:comparison}

In the table below, we give a more detailed comparison of our results to the 
Miller-Sheehy algorithm and its extensions. In Section~5.2 of their paper,
Miller and Sheehy describe a 
generalization of their approach that improves the running
time for small $d$ by 
computing higher order Tverberg points of depth $r$ 
by brute force. 
The approximation quality deteriorates by a factor of $r/2$. No 
exact bounds are given, but as far as we can tell, one can 
achieve a running time of $O(f(d) n^2)$ for fixed $d$ 
by setting the parameter $r = d + 1$, while losing a 
factor of $(d+1)/2$ in the approximation. 

Furthermore, even 
though it is not explicitly mentioned in their paper, we 
think that it is possible to also bootstrap the Miller-Sheehy algorithm 
(for a better running time in terms of $d$, while 
losing another factor of $(d+1)$ in the output). This is done
by performing the generalized procedure~\cite[Section~5.2]{MillerSh10} 
with $r = d+1$, but using the original Miller-Sheehy algorithm instead
of the brute-force algorithm.
Table~\ref{tab:summary} shows a rough comparison (ceilings omitted)
of the different approaches. Again, $f$ denotes the 
running time of the brute force algorithm.

\begin{center}
\begin{tabular}{| l | c | c |} \hline
  \textbf{Algorithm}	& \textbf{Running time} & \textbf{Depth} 
\tabularnewline\hline
  Theorem~\ref{thm:simpleAlgorithm} & $O(n)$ & $n/2^d$ 
\tabularnewline \hline  
  Miller-Sheehy 		    & $n^{O(\log d)}$ & $n/2(d+1)^2$ 
\tabularnewline\hline
  Theorem~\ref{thm:betterTverberg}  & $O\left( f(2^d) + d^{O(1)}n\right)$ & $n/2(d+1)^2$ 
\tabularnewline \hline
  Miller-Sheehy generalized ($r = d + 1$) & $O\left( f(d) n^2\right)$ & $\approx n/(d+1)^3$ 
\tabularnewline \hline
  Theorem~\ref{thm:betterTverberg-withMS} & $O\left(2^{O(d\log d)} + n\right)$ & $n/4(d+1)^3$
\tabularnewline \hline
  Miller-Sheehy bootstrapped	& $d^{O(\log d)} n^3$ & $\approx n/2(d+1)^4$ 
\tabularnewline \hline
  Theorem~\ref{thm:bootstrap}	& $d^{O(\log d)} n$ & $n/4(d+1)^3$ 
  \tabularnewline\hline
\end{tabular}\label{tab:summary}
\end{center}

We should emphasize that for all dimensions $d$ with $2^d \leq 2(d+1)^2$, 
i.e., $d \leq 7$, our simplest algorithm outperforms every other
approximation algorithm in both running time and approximation ratio. For 
example, it gives a $1/2$-approximate Tverberg point in $3$ dimensions in 
linear time. 

\section{Conclusion and Outlook}\label{sec:outlook}
We have presented a simple algorithm for finding an approximate 
Tverberg point. It runs in linear time for any fixed dimension. 
Using more sophisticated tools and combining our methods with known results, 
we managed to improve the running time to
$d^{O(\log d)} n$, while getting within a factor of
$1/4(d+1)^2$ of the bound from Tverberg's theorem.
Unfortunately, the resulting running time remains quasi\-polynomial in
$d$, and we still do not know whether there exists a polynomial algorithm
(in $n$ and $d$) for finding an approximate Tverberg point of linear depth.

However, we are hopeful that our techniques constitute a further step 
towards a truly polynomial time algorithm and that such an algorithm will eventually be 
discovered---maybe even by a more clever combination of 
our algorithm with that of Miller and Sheehy.
An alternative promising approach, suggested
to us by Don Sheehy, derives from a beautiful proof of Tverberg's 
theorem. It is due to Sarkaria and can be found in Matousek's
book~\cite[Chapter 8]{Mat02}. It uses the 
colorful Carath\'eodory theorem:
\begin{thm} [Colorful Carath\'eodory] 
Let $C_1 \uplus \dots \uplus C_{d+1} \subseteq \R^d$,
such that for $i = 1, \dots, d+1$, we have
$0 \in \ch(C_i)$. Then there is a set $C$ 
of $d+1$ points with $0 \in \ch(C)$ and $|C_i \cap C| = 1$.
\end{thm}
Sarkaria's proof transforms a $d$-dimensional instance of 
$n$ points of the Tverberg point problem to a Colorful Carath\'eodory 
problem in approximately $dn$ dimensions.

The question now is whether such a colorful simplex can 
be found in time polynomial in both $d$ and $n$. This 
would lead to a polynomial time algorithm for computing a 
Tverberg point. Observe that this would not contradict any 
complexity theoretic assumptions: an algorithm that \emph{finds}
such a point does not necessarily have to decide whether a given 
point indeed is a Tverberg point.

The simplest proof of Colorful Carath\'eodory leads directly 
to an algorithm for finding such a colorful simplex and works as 
follows: take an arbitrary colorful simplex. If the origin is not contained 
in it, delete the farthest color and take a point of that color that 
together with the other points induces a simplex that is closer to the
origin. It is unknown whether this procedure runs in polynomial time 
for both $d$ and $n$. Settling this question would constitute major
progress on the problem (see~\cite{MeunierDeza11,Rong12} for work in this direction).

Yet another approach would be to relax Sarkaria's proof and 
to try to formulate it as an approximation problem, which might be
easier to solve. However, it is not clear how to state such an approximation
to the Colorful Carath\'eodory problem in a way that leads to an approximate
Tverberg point. Perhaps via such a method, our algorithms 
can be improved further.

It is known that the problem of deciding whether a given point 
has at least a certain depth is NP-complete~\cite{Teng92}. It is possible to 
strengthen this result to show that in $\R^{d+1}$, the problem is 
$d$-\textsc{Sum} hard, using the approach by Knauer et al.~\cite{KTW11}. 
However, this does not tell us anything about the actual problem of 
computing a point of depth $n/(d+1)$. Such a point is guaranteed 
to exist, so it is not clear how to prove the problem  hard using 
``standard'' NP-completeness theory. Rather, we think that a 
hardness proof along the lines of complexity classes such as 
PPAD or PLS (see Papadimitriou~\cite{PPAD94}) should be pursued.

Finally, a common issue with Tverberg point (and centerpoint) algorithms in high 
dimensions, also pointed out by Clarkson~\etal~\cite{ClarksonEpMiStTe96}, is 
that the coefficients arising during the algorithm might become exponentially 
large.  While this is not a problem in our uniform cost model, for implementations 
of the algorithm it seems necessary to bound these. In particular, it would be
interesting to investigate the bit complexity of the intermediate solutions 
arising during the pruning process. As an alternative approach, one might try 
to perturb the points in the process, thereby lowering the precision of the 
coefficients. Additionally, one might have to introduce a notion of 
\emph{almost approximate Tverberg points}, where the point that is returned 
does not have to lie \emph{inside} all sets, but only \emph{close to} them.

\vskip0.9cm
\noindent\textbf{Acknowledgments.} We would like to thank
Nabil Mustafa for suggesting the problem to us. We also thank 
him and Don Sheehy for helpful
discussions and insightful suggestions. 

We would further like to thank the anonymous referees 
for their helpful and detailed comments.

\bibliographystyle{abbrv} 
\bibliography{tverberg}

\end{document}